\documentclass[11pt,a4paper]{article}
\usepackage{amsmath,amssymb,titling,authblk}
\usepackage{amsthm}
\usepackage[top=2.3cm,right=2.3cm,left=2.3cm,bottom=2.3cm]{geometry}
\usepackage{graphicx}
\usepackage{color} 
\usepackage{slashed}
\usepackage[pdfstartview=FitH,colorlinks=true,linkcolor=blue,anchorcolor=red,citecolor=magenta,urlcolor=blue]{hyperref}
\usepackage{amscd}
\usepackage[normalem]{ulem}
\usepackage{appendix}
\usepackage{bbold}
\usepackage{mathrsfs}
\usepackage{pdfsync}
\usepackage{bbm}
\usepackage{bm}
\usepackage[arrow,matrix,curve]{xy}
\usepackage{bbding}
\usepackage{wasysym}
\usepackage{booktabs}
\usepackage{siunitx}
\usepackage{cite}
\usepackage{epsf}
\usepackage{epsfig}
\usepackage{wrapfig}
\usepackage[utf8]{inputenc}
\usepackage{subcaption}

\DeclareCaptionFormat{custom}
{%
    \textbf{#1#2}\textit{\small #3}
}
\allowdisplaybreaks[4]
\numberwithin{equation}{section}

\definecolor{verde}{cmyk}{.83,.21,1,.08}
\definecolor{darkorchid}{rgb}{0.6, 0.2, 0.8}
\definecolor{darkgreen}{rgb}{0,.5,0}

\def\({\left(}
\def\){\right)}
\def\[{\left[}
\def\]{\right]}

\newcommand{\ii}{\mathrm{i}}

\newcommand{\dd}{\mathrm{d}}

\newcommand{\be}{\begin{equation}}
\newcommand{\ee}{\end{equation}}
\newcommand{\bea}{\begin{eqnarray}}
\newcommand{\eea}{\end{eqnarray}}
\newcommand{\del}{\partial}

\newcommand{\Tr}[1]{\:{\rm Tr}\,#1}
\newcommand{\la}{\label}

\newtheorem{proposition}{Proposition}[section]

\begin{document}

\title{Action principle for $\kappa$-Minkowski noncommutative $U(1)$ gauge theory from Lie-Poisson electrodynamics
}

\author[1,2]{M. A. Kurkov}

\affil[ ]{}

\affil[1]{\textit{\footnotesize Dipartimento di Fisica ``E. Pancini'', Universit\`a di Napoli Federico II, Complesso Universitario di Monte S. Angelo Edificio 6, via Cintia, 80126 Napoli, Italy.}}
\affil[2]{\textit{\footnotesize INFN-Sezione di Napoli, Complesso Universitario di Monte S. Angelo Edificio 6, via Cintia, 80126 Napoli, Italy.}}
\affil[ ]{}
\affil[ ]{\footnotesize e-mail: \texttt{max.kurkov@gmail.com}}
\maketitle
\begin{abstract}\noindent
Lie-Poisson electrodynamics describes a semiclassical approximation of noncommutative $U(1)$ gauge theories with Lie-algebra-type noncommutativities.
We obtain a gauge-invariant local classical action with the correct commutative limit for a generic Lie-Poisson gauge model, and present the corresponding deformed Maxwell equations. At the semiclassical level, our results provide a relatively simple solution to the old problem of constructing an admissible Lagrangian formulation for the $U(1)$ gauge theory on the four-dimensional $\kappa$-Minkowski space-time.  {On the one hand,} we derive an explicit expression for the classical action which yields the deformed Maxwell equations previously proposed in~\emph{JHEP 11 (2023) 200} for this noncommutativity on general grounds.  On the other hand, according to our analysis, these Maxwell equations follow from the action proposed in the present paper as the Euler-Lagrange equations for any Lie-algebra-type noncommutativity.
\end{abstract}


\section{ Introduction}
Noncommutative geometric structure of space-time~\cite{Doplicher:1994tu}, arising in different approaches to quantum gravity~\cite{SW, Perem}, alters the short-distance behavior of field-theoretical models. This fact motivates our interest in noncommutative field theory~\cite{Szabo}, in particular in noncommutative gauge theory~\cite{Wallet}.

Consider a manifold $\mathcal{M}\simeq\mathbb{R}^d$ representing space-time, equipped with a Kontsevich star product of smooth functions on it:
\be
f \star g = f \cdot g + \frac{\ii}{2} \{f,g\} + \cdots, \qquad f,g \in \mathcal{C}^{\infty}(\mathcal{M}), \la{sarp}
\ee
where $\{\,,\}$ stands for a given Poisson bracket on $\mathcal{M}$, and the remaining terms, denoted by dots, contain higher derivatives of  $f$ and $g$.

Many noncommutativities that have attracted significant attention in the literature, including  
the renowned $\kappa$-Minkowski case~\cite{lukierski,Majid:1994cy,Kosinski3,Dimitrijevic:2003pn,Meljanac:2007xb,Arzano:2009ci,Dimitrijevic:2005xw,ML1,Harikumar:2011um,Pachol, Juric:2018qdi,Mathieu:2020ccc,Lizzi:2020tci,Fabiano:2023xke}, are of the Lie-algebra type:
\be
[x^{\mu},x^{\nu}]_{\star} = \ii \,\mathcal{C}^{\mu\nu}_{\lambda} x^{\lambda}. \la{xcomrel}
\ee
Here $x^{\mu}$ denote the local coordinates on $\mathcal{M}$; the deformation parameters $\mathcal{C}^{\mu\nu}_{\lambda}$ are the structure constants of a given Lie algebra $\mathfrak{g}$; and the square brackets stand for the star-commutator:
\be
[f,g]_{\star} = f \star g - g \star f = \ii\,\{f,g\} +\cdots. \la{starc}
\ee
Throughout this article, we discuss Lie-algebra-type noncommutativities only, so our Poisson bracket  has the form:
\be
\{f,g\} = x^{\lambda} \, \mathcal{C}^{\mu\nu}_{\lambda} \, \partial_{\mu}f \,\partial_{\nu} g. \la{Pbre}
\ee

An important consequence of space-time noncommutatative geometry is a deformation of the gauge algebra. Consider two infinitesimal gauge transformations $\delta_f$ and $\delta_g$ of some dynamical variable, where the subscripts $f$ and $g$ indicate the corresponding gauge parameters. While in the usual $U(1)$ gauge theory these transformations commute, the noncommutative setting gives rise to the deformed non-Abelian algebra
\be
[\delta_f, \delta_g] = \delta_{-\ii[f,g]_{\star}}. \la{ncalg}
\ee
In the novel approach to noncommutative gauge theory proposed in~\cite{PatriziaVlad}, this relation is taken as a starting point, see also~\cite{Kupriyanov:2021cws}.

In the semiclassical approximation, which corresponds to slowly varying fields, the higher-derivative terms in~\eqref{starc} are negligible, and the algebra~\eqref{ncalg} reduces to the Poisson gauge algebra:
\be
[\delta_f, \delta_g] = \delta_{\{f,g\}}. \la{pga}
\ee
A deformation of the $U(1)$  theory, where the infinitesimal gauge transformations obey~\eqref{pga}, is called Lie-Poisson electrodynamics or Lie-Poisson gauge theory~\cite{our2023}. Being a field theory on a Poisson manifold, the Lie-Poisson gauge formalism provides the semiclassical limit of electrodynamics on Lie-algebra-type noncommutative space-time. 

We emphasise that this approximation is more than just a first-order correction in the deformation parameters $\mathcal{C}^{\mu\nu}_{\lambda}$. Although the semiclassical regime assumes that the fields vary slowly, it imposes no restrictions on their magnitudes. In particular, the product of the gauge field $A$ with $\mathcal{C}$ is not required to be small. While the algebra~\eqref{pga} contains only linear dependence on $\mathcal{C}$, as we shall see below, the gauge-covariant field strength and other constituents of Lie-Poisson electrodynamics involve \emph{all} orders in $\mathcal{C}$ (and in $A$), not just the leading one.

In recent years, Lie-Poisson electrodynamics has undergone rapid development~\cite{our2023,Kupriyanov:2019cug,Kurkov:2021kxa,Kupriyanov:2020axe,Abla:2022wfz,Kupriyanov:2021aet,Kupriyanov:2021cws,Kupriyanov:2022ohu,Kupriyanov:2023qot,DiCosmo:2023wth,Bascone:2024mxs,Kupriyanov:2024dny,Sharapov:2024bbu, Basilio:2024bir,Abla:2024wtr}. In the absence of matter, the deformed gauge transformations and the deformed Maxwell equations have been constructed for generic Lie-algebra-type noncommutativity~\cite{Kupriyanov:2021cws,our2023}. Charged point-like particles were studied in detail in~\cite{Kupriyanov:2024dny,Basilio:2024bir}. General prescriptions for the charged matter fields were outlined in~\cite{Sharapov:2024bbu}. The  present paper continues the research line of~\cite{our2023,Kupriyanov:2019cug,Kurkov:2021kxa,Kupriyanov:2020axe,Abla:2022wfz,Kupriyanov:2021aet,Kupriyanov:2021cws,Kupriyanov:2022ohu,Kupriyanov:2023qot,DiCosmo:2023wth,Bascone:2024mxs,Kupriyanov:2024dny,Sharapov:2024bbu, Basilio:2024bir,Abla:2024wtr}. Of course,  there are other studies on gauge theories on Poisson manifolds~\cite{Bimonte:1996fq, Arefeva:1999pkt, Jurco:2000fs, Meusburger:2003hc, Kennedy:2012gk, Cattaneo:2001bp}. A comparison of our approach with related ones can be found in~\cite{our2023}.  

Despite impressive progress, some important questions remain open. So far, an admissible Lagrangian formulation of the deformed Maxwell equations has been obtained for unimodular algebras $\mathfrak{g}$ only, that is, when the structure constants defining the noncommutativity~\eqref{xcomrel} obey the relation\footnote{
 {
The condition~\eqref{compcond} ensures that the linear map $\mathrm{ad}_{q} : \mathfrak{g} \longrightarrow \mathfrak{g}$ is traceless for all $q\in \mathfrak{g}$, that is, by definition, the Lie algebra $\mathfrak{g}$ is unimodular~\cite{Dufour}. In~\cite{Hersent,Mercati} unimodular Lie groups are considered in the  
noncommutative field-theoretical context. As explained in~\cite{Toronto}, the unimodularity of a Lie group implies the unimodularity of the corresponding Lie algebra; for connected Lie groups the converse is also true.
}}
\be
\mathcal{C}^{\mu\nu}_{\mu} = 0. \la{compcond}
\ee
While for the $\mathfrak{su}(2)$ and angular noncommutativities this condition is fulfilled, in the $\kappa$-Minkowski case it is not satisfied.

Technically, the problem is the following. The action $S_{\mathbf{g}}$ for the gauge field, defined simply as an integral over  space-time of a gauge-covariant Lagrangian density $\mathcal{L}$, 
\be
\delta_{f} \mathcal{L} = \{ \mathcal{L},f\},
\ee
is not necessarily gauge-invariant,
\be
\delta_{f} S_{\mathbf{g}} = \int_{\mathcal{M}}\, \dd x\, \{ \mathcal{L},f\} \neq 0.
\ee
In particular, when the equality~\eqref{compcond} does not hold, e.g. in the $\kappa$-Minkowski case, the Poisson bracket between two functions is not a total derivative and, therefore, its integral does not vanish~\cite{Kupriyanov:2020axe}. This property is nothing but the semiclassical version of the non-cyclicity of the corresponding $\star$-product,
\be
\int_{\mathcal{M}} \dd x\,\[ \,.\,,\, . \]_{\star}\, \neq 0, 
\ee
which blocks the development of gauge theories on the $\kappa$-Minkowski space.

There have been various attempts to overcome this difficulty, e.g., by inserting a measure $\mu(x)$ in the definition of the classical action in such a way that $\mu$ times the Poisson bracket becomes a total derivative~\cite{Kupriyanov:2020axe}. This proposal works iff the “compatibility condition”
\be
\partial_{\mu} \big(x^{\xi} \,\mathcal{C}_{\xi}^{\mu\nu}\,\mu(x)\big) = 0, \la{compatib}
\ee
is fulfilled. For the $\kappa$-Minkowski case, the most general solution of this system of partial differential equations was obtained in the quoted reference. In~\cite{our2023}, we have shown that, despite a large functional ambiguity, none of these solutions tends to 1 in the commutative limit. Therefore, the gauge-invariant classical action constructed along these lines is not a deformation of the usual Maxwell action and cannot be regarded as admissible.

Of course, some progress has been made in this direction. In~\cite{Wess}, a gauge-invariant action for a non-cyclic star-product was constructed in the two-dimensional case. In~\cite{Dimitrijevic:2005xw,ML1}, the problem was addressed in four dimensions, perturbatively up to linear order in the deformation parameter, using the Seiberg-Witten map.   {The former reference introduces a field-dependent volume factor, which allows one to overcome the above problem of the correct commutative limit of the measure $\mu$.}
An admissible five-dimensional action for the $\kappa$-Minkowski case was obtained in~\cite{Mathieu:2020ywc}. In~\cite{our2023}, in the semiclassical $\kappa$-Minkowski context, we proposed a one-parameter family of four-dimensional gauge-covariant field equations with the correct commutative limit and reasonable constraints generalising the Noether identity; however, the classical action was missing. In~\cite{Kupriyanov:2023qot}, again in the semiclassical context, general prescriptions for building the deformed Maxwell action were outlined for any Lie-algebra-type noncommutativity. The approach of~\cite{Kupriyanov:2023qot} exploits a gauge-covariant field strength~$F^{s}$, which transforms via a Lie derivative, whereas in~\cite{our2023} the field strength $\mathcal{F}$ transforms via a Poisson bracket (see the next section).  For the $\kappa$-Minkowski case Ref.~\cite{Kupriyanov:2023qot} provides explicit formulae for the main constituents of the symplectic groupoid approach, namely the source and the target maps. However, any analysis of the Euler-Lagrange dynamics, and in particular its comparison with the deformed Maxwell equations proposed earlier in~\cite{our2023}, is absent there.

To the best of our knowledge, an explicit expression for a deformed Maxwell action describing the semiclassical regime of four-dimensional noncommutative electrodynamics in the $\kappa$-Minkowski space, has not been obtained so far. The present paper aims to fill this gap by providing a Lagrangian formulation of the field equations from the one-parameter family proposed in~\cite{our2023} on general grounds. Moreover, we develop a Lagrangian description of the deformed Maxwell dynamics for arbitrary Lie-algebra-type noncommutativity, thereby extending the admissible Lagrangian analysis of~\cite{our2023}, applicable to the unimodular case only.

In Sec.~\ref{review}, we introduce the main building blocks of the Lie-Poisson gauge formalism, following Ref.~\cite{our2023} and references therein. The most important section of this article is Sec.~\ref{action}, where we construct a local gauge-invariant classical action for \emph{any} Lie-algebra-type noncommutativity  {and derive the Euler-Lagrange field equations}. In Sec.~\ref{KappaSec}, we apply our findings to the four-dimensional $\kappa$-Minkowski case. 

\section{Lie-Poisson electrodynamics: building blocks}\label{review}
From a technical point of view, the basic elements of Lie-Poisson electrodynamics are two $d \times d$ matrices $\gamma$ and $\rho$, which depend on the gauge field $A_{\mu}(x)$. By definition, $\gamma(A)$ and $\rho(A)$ solve the \emph{master} equations\footnote{Throughout this paper, for any matrix the upper index enumerates rows, while the lower index enumerates columns.}
\be
\gamma_{ \mu }^{\nu} (A) \,\frac{\partial\gamma^{ \xi}_{ \lambda}(A)}{\partial A_{\mu}} 
- \gamma^{ \xi}_{ \mu }(A)\,\frac{\partial \gamma^{ \nu }_{ \lambda} (A)}{\partial A_{\mu}}
  = \mathcal{C}_{ \mu}^{ \nu  \xi}\,\gamma^{ \mu }_{ l}(A), \qquad \gamma^{\nu}_{\lambda}(A)\, \frac{\partial\rho_{\xi}^{\mu}(A)}{\partial A_{\lambda}} \,+ \rho_{\xi}^{\lambda} (A)\,\frac{\partial\gamma_{\lambda}^{\nu}(A)}{\partial A_{\mu}}   = 0, \la{master}
\ee
and they tend to the identity matrices in the commutative limit of vanishing structure constants:
\be
\lim_{\mathcal{C}\to 0} \gamma^{\mu}_{\nu}(A)   = \delta^{\mu}_{\nu}, \qquad \lim_{\mathcal{C}\to 0} \rho^{\mu}_{\nu}(A) = \delta^{\mu}_{\nu}. \la{grcomlim}
\ee  

The main constituents of the formalism, namely the deformed gauge transformations $\delta_f A_{\mu}$, which close the Poisson algebra~\eqref{pga}, the deformed field strength $\mathcal{F}_{\mu\nu}$, and the deformed gauge-covariant derivative $\mathcal{D}_{\mu}$, are constructed in terms of $\gamma$ and $\rho$ as follows:
\bea
\delta_f A_{\mu} &:=&\gamma^{\xi}_{\mu}(A)\,\partial_{\xi} f(x) +\{A_{\mu} ,f\}, \nonumber\\
\mathcal{F}_{\mu\nu}(x) &:=& \rho_{\mu}^{\xi}(A)\,\rho_{\nu}^{\lambda}(A)\big(\gamma_{\xi}^{\sigma}(A)\,\partial_{\sigma}A_{\lambda}
   -\gamma_{\lambda}^{\sigma}(A)\,\partial_{\sigma}A_{\xi}+\{A_{\xi},A_{\lambda}\}\big), \nonumber\\
   \mathcal {D}_{\mu}\psi(x)&:=&\rho_{\mu}^{\nu}(A)\, (\gamma_{\nu}^{\xi}(A) \del_{\xi} \psi+\{A_{\nu},\psi\}) . \la{techel}
\eea
In the last line, $\psi$ denotes an arbitrary field\footnote{ {To avoid confusion, we emphasise that the notation $\psi$ is used in the above sense only and is not related to the spinorial field of the usual electrodynamics.}}  transforming in a covariant manner under gauge transformations,  $\delta_f \psi = \{ \psi, f\}$,  {e.g., the deformed field strength $\mathcal{F}_{\nu\lambda}$.
}

The master equations~\eqref{master}, together with the requirements~\eqref{grcomlim}, yield the desired properties:
\be
[\delta_f, \delta_g] A_{\mu} = \delta_{\{f,g\}} A_{\mu},\qquad\delta_f {\cal F}_{\mu\nu}=\{{\cal F}_{\mu\nu},f\},\qquad \delta_f\left({\cal D}_{\mu}\psi\right)=\{{\cal D}_{\mu}\psi,f\}, \la{trarul}
\ee 
and correct commutative limits\footnote{ {In the commutative limit, the Poisson bracket on the space-time $\mathcal{M}$ vanishes, thus any field, transforming in a gauge-covariant manner $\delta_f \psi = \{ \psi, f\}$, becomes gauge-invariant: $\lim_{\mathcal{C}\to0}\delta_f\psi = 0$. Therefore at $\mathcal{C}\to 0$ the gauge-covariant derivative $\mathcal{D}$ must reduce to the usual partial derivative $\partial$. }}:
\be
\lim_{\mathcal{C}\to 0} \delta_f A_{\mu} = \partial_{\mu}f, \qquad  \lim_{\mathcal{C}\to 0} \mathcal{F}_{\mu\nu} = F_{\mu\nu}, \qquad  \lim_{\mathcal{C}\to 0}\mathcal {D}_{\mu}\psi  = \partial_{\mu}\psi, \la{goodcomlims}
\ee
where
\be
F_{\mu\nu} = \partial_{\mu} A_{\nu} - \partial_{\nu} A_{\mu}
\ee
is the usual Abelian field-strength.

The ``universal” solutions of the master equations~\eqref{master}, which are valid for any Lie-algebra-type noncommutativity, can be constructed explicitly in terms of matrix-valued functions~\cite{Kupriyanov:2021cws,Kupriyanov:2022ohu}:
\be
\gamma_{\mathbf{u}}(A) = G(\hat{A}), \qquad \rho_{\mathbf{u}}(A) = \frac{1}{G(-\hat{A})}, \qquad \hat{A}^{\mu}_{\nu} :=   \mathcal{C}^{\sigma\mu}_{\nu} A_{\sigma},  \la{unive}
\ee
 {where the form factors are given by
\bea
G(s) = \frac{s}{2} + \frac{s}{2}\,\mathrm{coth}\,\frac{s}{2} = \sum_{k=0}^{\infty} \frac{B_k^{+}\,s^k}{k!} , \qquad
\frac{1}{G(-s)} = \frac{\exp{(s)}-1}{s} = \sum_{k=0}^{\infty} \frac{s^k}{(k+1)!},
\la{Gdef}
\eea
with $B_k^{+}$, $k =0,1,2,...$, being the Bernoulli numbers, and the subscript ``${\mathbf{u}}$'' stands for ``universal''.}

Remarkably, any invertible field redefinition that reduces to the identity map in the commutative limit,
\be
A_{\mu}(x)\longrightarrow  \tilde{A}_{\mu}\big(A(x)\big), \qquad \lim_{\mathcal{C}\to 0}\tilde{A}_{\mu}(A) = A_{\mu}(x), \la{SW}
\ee
generates new admissible solutions of the master equations~\eqref{master}:
\be
{\gamma}^{\mu}_{\nu} ({A}) = \Bigg(\big[\gamma_{\mathbf{u}}\big]_{\xi}^{\mu}(\tilde{A})\cdot \frac{\partial {A}_{\nu}}{\partial \tilde{A}_{\xi}}\Bigg)\Bigg|_{\tilde{A} = \tilde{A}({A})}, \quad {\rho}_{\nu}^{\mu}({A})  = \Bigg(\frac{\partial \tilde A_{\xi}}{\partial {A}_{\mu}}\cdot \big[\rho_{\mathbf{u}}\big]_{\nu}^{\xi}(\tilde{A})\Bigg)\Bigg|_{\tilde{A} = \tilde{A}({A})}. \la{newGR}
\ee
From now on, we shall assume that $\gamma$ and $\rho$ are either given by the “universal” expressions~\eqref{unive} or can be obtained from them through the relations~\eqref{SW} and~\eqref{newGR}. The field-theoretical models based on the former and the latter choices of $\gamma$ and $\rho$ will be referred to as the “universal” and “universal-equivalent” realizations of Lie-Poisson electrodynamics, respectively.
  \\

{{\noindent \footnotesize{{\bf Remark.} \emph{The differential-geometric meaning of $\gamma$ and $\rho$ was clarified in~\cite{Kupriyanov:2023qot,Kupriyanov:2024dny,Sharapov:2024bbu}. Let $G$ be a Lie group whose Lie algebra is $\mathfrak{g}$, and let $p_{\mu}$ be local coordinates on $G$. Then  $\gamma^{\mu}_{\nu}(A)$ and $\rho^{\mu}_{\nu}(A)$ are the local components of the left-invariant vector fields
\be
\gamma^{\mu}(p) = \gamma_{\nu}^{\mu}(p)\,\frac{\partial}{\partial p_{\nu}}, \qquad \mu =0,...,d-1,
\ee
and the right-invariant one-forms
\be
\rho_{\mu}(p) = \rho^{\nu}_{\mu}(p)\,\dd p_{\nu}, \qquad \mu =0,...,d-1,
\ee
on $G$, respectively, evaluated at $p_{\mu} = A_{\mu}(x)$. From this point of view, the field redefinition~\eqref{SW} corresponds to a change of local coordinates $p_{\mu} \longrightarrow \tilde p_{\mu}(p)$ on $G$. The components of the vector fields $\gamma^{\mu}$ and one-forms $\rho_{\mu}$ transform accordingly, cf. Eq.~\eqref{newGR}.}}} \\

By using the basic notions~\eqref{techel}, in the next section we shall construct a simple expression for the gauge-invariant classical action which is valid for any Lie-algebra-type noncommutativity.

\section{Gauge-invariant classical action} \la{action}
The main idea is to introduce an $A$-dependent integrating factor $M_{A}(x)$ which converts gauge-covariant expressions into gauge-invariant ones (up to a total derivative).  First, we construct it and after that, we shall focus on the action.
\subsection*{a. Integrating factor}
Let us define $M_A(x)$ as
\be
M_{A}(x) := \Big(\det{\big[\gamma{\big(A(x)\big)}\,\rho\big(A(x)\big)\big]}\Big)^{-1}\, . \la{mudef}
\ee 
Eq.~\eqref{grcomlim} implies the correct commutative limit:
\be
 \lim_{\mathcal{C}\to 0}M_{A}(x) = 1. \la{muAcomlim}
\ee
In order to prove that it is indeed the required integrating factor (see Proposition~\ref{prop3} below), we shall calculate this object explicitly and study its transformation properties.
\begin{proposition}\label{prop1}
Explicit expressions for $M_{A}(x)$ read:
\begin{itemize} 
\item{For the universal realization of Lie-Poisson electrodynamics,
\be
M_{A}(x) = \exp{\left(\mathcal{C}_{\mu}^{\mu\sigma}A_{\sigma}(x)\right)}; \la{mufirs}
\ee
}
\item{ For the universal-equivalent realization of Lie-Poisson electrodynamics,
\be
M_{A}(x) = \exp{\left(\mathcal{C}_{\mu}^{\mu\sigma}\tilde{A}_{\sigma}(A(x))\right)}. \la{musec}
\ee
}
\end{itemize}
\end{proposition}
\begin{proof}
To prove the first statement, we notice that the form factor $G(s)$, defining the universal expressions~\eqref{unive}, satisfies the algebraic identity
\be
\frac{G(s)}{G(-s)} = \exp{(s)},
\ee 
therefore, for the matrices $\gamma_{\mathbf{u}}$ and $\rho_{\mathbf{u}}$ we have
\be
\gamma_{\mathbf{u}}(A) \rho_{\mathbf{u}}(A) = \frac{G(\hat{A})}{G(-\hat{A})}  = \exp{(\hat{A})}.
\ee
Consequently, by using the relation $\ln \det = \Tr \ln$, we find
\bea 
 \det{\big[\gamma_{\mathbf{u}}(A) \rho_{\mathbf{u}}(A)\big]}= \det{\Big[\exp{\big(\hat{A}(x)\big)}\Big]} = \exp{\big( \Tr{\hat{A}(x) }\big)} =\exp{\left(\mathcal{C}_{\mu}^{\sigma\mu}A_{\sigma}(x)\right)}. \la{interme}
\eea
Substituting this expression into the definition~\eqref{mudef} of  $M_{A}$, and using the skew-symmetry of the structure constants in the upper indices, we arrive at the desired relation~\eqref{mufirs}.

To prove the second statement, we notice that under the field redefinition~\eqref{SW}, the relation~\eqref{newGR} implies
\be
 \gamma_{\xi}^{\mu}(A)\, \rho_{\nu}^{\xi}(A) =\gamma_{\xi}^{\mu}(A)\cdot \underbrace{
 \frac{\partial \tilde{A}_{\sigma}}{\partial A_{\xi}}  \frac{\partial A_{\lambda}}{\partial \tilde{A}_{\sigma}}
 }_{\delta^{\xi}_{\lambda}}
 \cdot \rho_{\nu}^{\lambda}(A)  =\Big(\big[\gamma_{\mathbf{u}}\big]^{\mu}_{\sigma} (\tilde{A}) \,  \big[\rho_{\mathbf{u}}\big]_{\nu}^{\sigma}(\tilde{A})  \Big)\Big|_{\tilde{A} = \tilde{A}(A)} . \la{secsteinterme}
\ee 
Therefore,
\be
 \det{\big[\gamma(A) \rho(A)\big]} =  \Big(\det{\big[\gamma_{\mathbf{u}} (\tilde{A}) \, \rho_{\mathbf{u}}(\tilde{A})\big]}\Big)\Big|_{\tilde{A} = \tilde{A}(A)} 
 = \exp{\left( {\mathcal{C}_{\mu}^{\sigma\mu}}\tilde{A}_{\sigma}(A(x))\right)}, \la{interme2}
\ee
where we used the identity~\eqref{interme} at the last step. Substituting this formula into the definition~\eqref{mudef}, we immediately obtain Eq.~\eqref{musec}.
\end{proof}

\noindent Now we discuss the transformation properties of $M_{A}(x)$.
\begin{proposition}\label{prop2}
Upon the deformed gauge transformations, the expression~\eqref{mudef} transforms as follows:
\be
\delta_fM_{A} =  M_{A} \, \mathcal{C}_{\nu}^{\nu\sigma} \,\partial_{\sigma}f  + \{ M_{A} , f\}. \la{prop2rel}
\ee
\end{proposition}
\begin{proof}
First, we prove the proposition for the universal realization of Lie-Poisson electrodynamics.
The explicit formulae~\eqref{mufirs} for $M_A$ and~\eqref{techel} for $\delta_f A$ yield:
\bea
\delta_{f} M_{A} 
 =\frac{\partial {M_{A}}}{\partial A_{\xi}} \,\,\delta_{f}A_{\xi} =  M_{A} \, \mathcal{C}_{\nu}^{\nu\xi} \,\big[\gamma_{\mathbf{u}}\big]_{\xi}^{\sigma}\,\partial_{\sigma}f  + \{ M_{A} , f\}.
\eea
To complete the proof, it is sufficient to demonstrate that
\be
 \mathcal{C}_{\nu}^{\nu\xi} \,\big[\gamma_{\mathbf{u}}\big]_{\xi}^{\sigma}
 =    \mathcal{C}_{\nu}^{\nu\sigma} . \la{todemost}
\ee
By contracting the Jacobi identity for the structure constants,
\be
 \mathcal{C}_{\alpha}^{\nu\sigma}\,\mathcal{C}^{\beta\xi}_{\sigma} + 
 \mathcal{C}_{\alpha}^{\beta\sigma}\,\mathcal{C}^{\xi\nu}_{\sigma} +  
 \mathcal{C}_{\alpha}^{\xi\sigma}\,\mathcal{C}^{\nu\beta}_{\sigma} = 0
\ee
over the indices $\nu$ and $\alpha$, we obtain
\be
 \mathcal{C}_{\nu}^{\nu\sigma}\,\mathcal{C}^{\beta\xi}_{\sigma} = 0, 
\ee
and therefore
\be
 \mathcal{C}_{\nu}^{\nu\sigma}\,\big[\hat{A}^k\big]^{\xi}_{\sigma} = 0, \qquad \forall k=1,2,... .  \la{interme3}
\ee
 {By substituting }
\be
\big[\gamma_{\mathbf{u}}\big]^{\xi}_{\sigma}(A) = \delta^{\xi}_{\sigma}+ \sum_{k=1}^{\infty}\frac{\big[\hat{A}^k\big]^{\xi}_{\sigma} B_k^{+}}{k!}, \la{gammaexp}
\ee
into the left-hand side of~\eqref{todemost}, we see that, thanks to the relation~\eqref{interme3}, the contributions of all nonzero powers of $\hat{A}$ vanish, while the Kronecker symbol gives the desired right-hand side of~\eqref{todemost}, what completes our proof for the “universal” realization.

For the universal-equivalent realization of Lie-Poisson electrodynamics, the relation~\eqref{prop2rel} can also be easily proven. Indeed, the explicit expression~\eqref{musec} for $M_{A}$ yields:
\bea
\delta_{f} M_{A}  &=&
 M_{A} \, \mathcal{C}_{\nu}^{\nu\sigma} \, \frac{\partial \tilde{A_\sigma} }{\partial  A_{\lambda}}\, \gamma_{\lambda}^{\xi}(A)\,\partial_{\xi} f  + \{ M_{A} , f\} \nonumber\\
  &=&  M_{A}\,\mathcal{C}_{\nu}^{\nu\sigma} \, \left[\gamma_{\mathbf{u}}\right]_{\sigma}^{\xi}(\tilde{A}(A))\,\partial_{\xi} f  + \{ M_{A} , f\}  = 
  M_{A}\, \mathcal{C}_{\nu}^{\nu\xi} \,\partial_{\xi} f + \{ M_{A} , f\} ,
\eea
where we used the identity~\eqref{todemost} for $\gamma_{\mathbf{u}}$ at the last step.
\end{proof}

The key property of $M_{A}(x)$ is established in the following proposition.
\begin{proposition}\label{prop3}
For any  quantity $\mathcal{Q}(x)$, transforming in a gauge-covariant way
\be
\delta_{f}\mathcal{Q} = \{\mathcal{Q},f\}, \la{Qtra}
\ee
the  expression
\be
\check{\mathcal{Q}}(x) := M_{A}(x) \,\mathcal{Q}(x) \la{Qinv}
\ee
is gauge-invariant up to a total derivative:
\be
\delta_{f}\check{\mathcal{Q}}(x) = \partial_{\nu}\big( x^{\xi}\, \mathcal{C}_{\xi}^{\nu\sigma}\,\check{\mathcal{Q}}(x) \,\partial_{\sigma}f\big). \la{QinvTra}
\ee
\end{proposition}
\begin{proof}
According to the transformation law~\eqref{Qtra} and Proposition~\ref{prop2}, upon an infinitesimal gauge transformation the quantity~\eqref{Qinv} transforms as follows:
\bea
\delta_f\check{\mathcal{Q}}(x) &=&\mathcal{Q}(x)\, \delta_f  M_{A}(x) +  M_{A}(x) \, \delta_f\mathcal{Q}(x) \nonumber\\
&=& \check{\mathcal{Q}}(x)\, \mathcal{C}_{\nu}^{\nu\sigma} \,\partial_{\sigma}f +\{\check{\mathcal{Q}}(x) ,f\}. \la{intermef1}
\eea
Using the explicit expression~\eqref{Pbre} for the Poisson bracket, we can rewrite the second term of the last line as:
\bea
\{\check{\mathcal{Q}}(x) ,f\} &=& x^{\xi}\, \mathcal{C}_{\xi}^{\nu\sigma}\,\partial_{\nu}\check{\mathcal{Q}}(x) \,\partial_{\sigma}f
\nonumber\\
 &=& -  \check{\mathcal{Q}}(x) \, \mathcal{C}_{\nu}^{\nu\sigma} \,\partial_{\sigma}f +\partial_{\nu}\big( x^{\xi}\, \mathcal{C}_{\xi}^{\nu\sigma}\,
\check{\mathcal{Q}}(x) \,\partial_{\sigma}f\big),  \la{intermef2}
\eea
what immediately implies~\eqref{QinvTra}. 
\end{proof}

\subsection*{b. Gauge-invariant action}
In what follows, we shall use the flat Minkowski metric
\be
\eta = \mathrm{diag}\,(\,+1,-1,-1,-1),
\ee
to raise and lower the indices; for instance,
\be
\mathcal{F}^{\mu\nu}(x) = \eta^{\mu\alpha}\eta^{\nu\beta}\,\mathcal{F}_{\alpha\beta}(x).
\ee
Now we are ready to prove our main result.
\begin{proposition}\label{prop4}
The classical action
\be
S_{\mathbf{g}}[A] = \int_{\mathcal{M}} \, \dd x\, \check{\mathcal{L}}(x) \la{invaction}
\ee
with the Lagrangian density
\be
\check{\mathcal{L}}(x) = M_{A}(x) \Big(-\frac{1}{4} \,\mathcal{F}_{\mu\nu}(x)\,\mathcal{F}^{\mu\nu}(x)\Big),  \la{Lmin}
\ee
\begin{itemize}
\item{
is gauge-invariant:
\be
\delta_{f} S_{\mathbf{g}}[A] = 0,
\ee
}
\item{
and has the correct commutative limit:
\be
\lim_{\mathcal{C}\to 0 }S_{\mathbf{g}}[A] =  \int_{\mathcal{M}} \, \dd x\,\Big( -\frac{1}{4}\,F_{\mu\nu}(x)\, F^{\mu\nu}(x)\Big) . \la{Scorrcomlim}
\ee
}
\end{itemize}

\end{proposition}
\begin{proof}
The correct commutative limit~\eqref{Scorrcomlim} is an obvious consequence of the commutative limits~\eqref{muAcomlim} for $M_A$ and~\eqref{goodcomlims} for $\mathcal{F}$; thus, from now on, we shall focus on gauge invariance.

The transformation law~\eqref{trarul} for $\mathcal{F}$, along with Leibniz’s rule for the Poisson bracket,
\be
\{f\,g,q\} = \{f,q\}\, g +   f\,\{g,q\}, \qquad \forall f,g,q\in\mathcal{C}^{\infty}(\mathcal{M})
\ee
implies that the expression
\be
\mathcal{L}(x) := -\frac{1}{4} \,\mathcal{F}_{\mu\nu}(x)\,\mathcal{F}^{\mu\nu}(x)
\ee
 transforms in a gauge-covariant way:
\be
\delta_{f}\mathcal{L} = \{\mathcal{L},f\}. \la{Ltra}
\ee
Since
\be
\check{\mathcal{L}}(x)= M_A(x)\,\mathcal{L}(x),
\ee
Proposition~\ref{prop3} states that
\be
\delta_{f}\check{\mathcal{L}}(x) = \partial_{\nu}\big( x^{\xi}\, \mathcal{C}_{\xi}^{\nu\sigma}\,\check{\mathcal{L}}(x)\,\partial_{\sigma}f\big)
\ee
Being an integral of a total derivative, the corresponding variation of the action vanishes, provided the gauge field $A(x)$ decays sufficiently fast at infinity:
\be
\delta_f S_{\mathbf{g}}[A] 
 = \int_{\mathcal{M}}\dd x\, \partial_{\nu}\big( x^{\xi}\, \mathcal{C}_{\xi}^{\nu\sigma}\,\check{\mathcal{L}}(x) \,\partial_{\sigma}f\big) =0.
\ee
\end{proof}

The gauge-invariance of the action~\eqref{invaction} implies that the corresponding Euler-Lagrange equations  \\
\be
\mathcal{E}_{EL}^{\mu}(x) = 0, \qquad \mathcal{E}_{EL}^{\mu}(x):= \frac{\delta S_{\mathbf{g}}[A]}{\delta A_{\mu}(x)}, \la{dMe}
\ee
are not independent but obey the Noether identity, established in the following proposition.
\begin{proposition} \la{prop5}
The left-hand sides $\mathcal{E}_{EL}^{\mu}(x)$ of the field equations~\eqref{dMe} obey the relation
\be
\partial_{\nu}\big(\gamma_{\mu}^{\nu}(A)\,\mathcal{E}_{EL}^{\mu}(x)\big) 
+ \{A_{\mu},\mathcal{E}_{EL}^{\mu}(x) \} + \mathcal{C}_{\nu}^{\xi\nu}\, \partial_{\xi}A_{\mu}\,\mathcal{E}_{EL}^{\mu}(x) = 0.  \la{gind}
\ee
\end{proposition}
\begin{proof}
Presenting the gauge transformation $\delta_{f}A$ in the form
\be
\delta_{f} {A_{\mu}} = \big(\gamma_{\mu}^{\nu}(A) + x^{\sigma}\,\mathcal{C}_{\sigma}^{\xi\nu} \partial_{\xi}A_{\mu}\big) \,\partial_{\nu}f, 
\ee 
we see that the gauge invariance of $S_{\mathrm{g}}$ yields
\bea 
0&=&\delta_{f} S_{\mathrm{g}}[A] = \int_{\mathcal{M}}\dd x\,\, \frac{\delta S_{\mathbf{g}}[A]}{\delta A_{\mu}(x)} \, \,\delta_{f}A_{\mu}(x) \nonumber\\
 &=& \int_{\mathcal{M}}\dd x\,\mathcal{E}_{EL}^{\mu}(x) \,\big(\gamma_{\mu}^{\nu}(A) + x^{\sigma}\,\mathcal{C}_{\sigma}^{\xi\nu} \partial_{\xi}A_{\mu}\big)\,\partial_{\nu}f \nonumber\\
 &=& -  \int_{\mathcal{M}}\dd x \,\big[\partial_{\nu}\big(\gamma_{\mu}^{\nu}(A)\,\mathcal{E}_{ {EL}}^{\mu}(x)\big) 
+ \{A_{\mu},\mathcal{E}_{EL}^{\mu}(x) \} + \mathcal{C}_{\nu}^{\xi\nu}\, \partial_{\xi}A_{\mu}\,\mathcal{E}_{EL}^{\mu}(x)\big]\, f, \la{ipropo}
\eea
where we integrated by parts at the last step.  Since~\eqref{ipropo} is valid for any gauge parameter $f(x)$, the expression in the square brackets must vanish identically, what implies Eq.~\eqref{gind}.
\end{proof}

Before we illustrate our findings with the four-dimensional $\kappa$-Minkowski example, we would like to make a few general remarks. \\

\noindent {\bf Remark~1.} Eq.~\eqref{Lmin} provides the minimal choice of the deformed Lagrangian density.  By contracting the structure constants $\mathcal{C}^{\mu\nu}_{\alpha}$ with the deformed field strength $\mathcal{F}_{\mu\nu}$ and the gauge-covariant derivative $\mathcal{D}_{\mu}$ we obtain new gauge-covariant expressions. Moreover, the Poisson bracket of two gauge-covariant quantities is again a gauge-covariant quantity. 
Multiplying these gauge-covariant combinations by the integrating factor $M_A$, we can easily construct many other admissible, though non-minimal, Lagrangian densities. \\

\noindent {\bf Remark~2. }For an unimodular Lie algebra $\mathfrak{g}$, the relation~\eqref{compcond}, together with Proposition~\ref{prop1}, implies that
\be
M_{A}(x) = 1,
\ee
and the expression~\eqref{invaction} reduces to the “admissible” gauge-invariant action previously proposed in~\cite{our2023}. \\

\noindent {\bf Remark~3. }For any Lie-algebra-type noncommutativity, a local first-order action $S_{\mathbf{particle}}$, which describes the motion of a point-like particle in a given gauge background $A$, was constructed in~\cite{Basilio:2024bir}. This action is invariant under the gauge transformations $\delta_f A$ of the background field $A$, accompanied by transformations of the phase-space variables of the charged particle, which close the gauge algebra~\eqref{pga}.

Now consider $N$ charged particles interacting with the gauge field. Combining the results of the present paper with those of~\cite{Basilio:2024bir}, we obtain a total gauge-invariant action describing the dynamics of this system:
 \be
S_{\mathbf{total}} = S_{\mathbf{g}} + \sum^{N}_{i=1}S^{(i)}_{\mathbf{particle}}, \la{fieldluspart}
\ee
with $S^{(i)}_{\mathbf{particle}}$ being the action of the $i$-th particle.  {Eq.~\eqref{fieldluspart} provides a deformation of the usual action describing a system of charged point-like particles and the electromagnetic field, interacting with them, see e.g. Eq. (4.11) of~\cite{Lechner}. Note that in the Lie-Poisson approach it is more convenient to describe the particle degrees of freedom in terms of the phase-space variables.

The first-order action $S^{(i)}_{\mathbf{particle}}$ depends on the trajectory of the $i$-th particle $\big(x_{(i)}(\tau), p_{(i)}(\tau)\big)$ in phase space, with $\tau$ being a parameter\footnote{ {An additional `gauge' freedom of reparametrisations of particle's trajectory, the corresponding Lagrange multipliers and so on are discussed in~\cite{Basilio:2024bir}.} } along this trajectory, 
and on the gauge background $A(x_{(i)}(\tau))$ evaluated at the particle's position in configuration space.  By varying the action~\eqref{fieldluspart} with respect to the particle variables $x_{(i)}(\tau)$ and $p_{(i)}(\tau)$, one obtains Hamiltonian equations of motion describing the particle dynamics in a given gauge background~\cite{Basilio:2024bir}.
On the other hand, by varying~\eqref{fieldluspart} with respect to $A(x)$, one obtains the deformed field equations for given particle trajectories  
in phase space. In particular, the variational derivative of $S^{(i)}_{\mathbf{particle}}$ over $A(x)$ yields a corresponding contribution to the current density. By solving the deformed Maxwell equations in the presence of this current density, one may analyze, e.g., electromagnetic radiation produced by moving particles in Lie-Poisson electrodynamics.  This analysis goes beyond the scope of the present paper, and below we shall focus on the Maxwell equations in the absence of charged matter.}

 {
\subsection*{c. Deformed Maxwell equations.}
The following proposition provides an explicit form of the deformed Maxwell equations.
}
\begin{proposition} \la{propMaxwell}
The Euler-Lagrange equations~\eqref{dMe} admit the manifestly gauge-covariant form\footnote{That is, $\delta_{f} \mathcal{E}^{\mu}_{G}  = \{\mathcal{E}^{\mu}_{G} ,f\}$.}:
\be
\mathcal{E}^{\mu}_{G} (x)= 0, \la{coveq}
\ee
with
\be
 \mathcal{E}_{G}^{\mu}(x) := \mathcal{D}_{\xi} \mathcal{F}^{\xi\mu} + \frac{1}{2}\, \mathcal{F}_{\lambda\omega}\,\mathcal{C}_{\nu}^{\lambda\omega}\, \mathcal{F}^{\mu\nu} 
 -\mathcal{F}_{\lambda \omega} \, \mathcal{C}_{\nu}^{\mu \omega}\,\mathcal{F}^{\lambda \nu} 
  - \frac{1}{4}\,\big( \mathcal{C}_{\nu}^{\nu \mu}\, \mathcal{F}_{\lambda\omega}\,\mathcal{F}^{\lambda\omega} 
+ 4\,\mathcal{C}_{\nu}^{\nu \lambda}\,\mathcal{F}_{\lambda \omega}\,\mathcal{F}^{\omega \mu}\big).
\ee
\end{proposition}
 {\begin{proof}
By calculating the left-hand sides of the Euler-Lagrange equations in the standard way, we obtain
\bea
\mathcal{E}_{EL}^{\omega} &=& \frac{\partial \check{\mathcal{L}}}{\partial A_{\omega}} - \frac{\partial }{\partial x^{\sigma}}\,\frac{\partial \check{\mathcal{L}}}{\partial (\partial_{\sigma}A_{\omega})}\nonumber\\
&=& M_A \, \bigg(\frac{\partial {\mathcal{L}}}{\partial A_{\omega}} - \frac{\partial }{\partial x^{\sigma}}\,\frac{\partial{\mathcal{L}}}{\partial (\partial_{\sigma}A_{\omega})} \bigg)
+ \frac{\partial M_A}{\partial A_{\omega}}\cdot\mathcal{L} - \frac{\partial M_A}{\partial x^{\sigma}}\frac{\partial \mathcal{L}}{\partial (\partial_{\sigma} A_{\omega})}. \la{ELderive1}
\eea
In Appendix~\ref{appA}, we prove the equalities
\bea
\frac{\partial {\mathcal{L}}}{\partial A_{\omega}} - \frac{\partial }{\partial x^{\sigma}}\,\frac{\partial{\mathcal{L}}}{\partial (\partial_{\sigma}A_{\omega})} &=&
\rho_{\nu}^{\omega}(A)\,\Big(\mathcal{D}_{\mu}\mathcal{F}^{\mu\nu}+\frac{1}{2}\,\mathcal{C}_{\mu}^{\sigma\varepsilon} \,\mathcal{F}_{\sigma\varepsilon}\,
\mathcal{F}^{\nu\mu}- \mathcal{C}_{\mu}^{\nu \varepsilon}\,\mathcal{F}_{\lambda \varepsilon} \, \mathcal{F}^{\lambda \mu} \Big) \nonumber\\
&-& \rho_{\nu}^{\omega}(A) \,\rho_{\mu}^{\alpha}(A) \,\mathcal{C}_{\sigma}^{\sigma\xi}\,\mathcal{F}^{\mu\nu} \,\partial_{\xi}A_{\alpha},  \la{appEq1}
\eea
and
\be
\frac{\partial M_A}{\partial A_{\omega}}\cdot\mathcal{L} - \frac{\partial M_A}{\partial x^{\sigma}}\frac{\partial \mathcal{L}}{\partial (\partial_{\sigma} A_{\omega})}
=  M_A\,\rho_{\nu}^{\omega} \,\rho_{\mu}^{\alpha} \,\mathcal{C}_{\sigma}^{\sigma\xi}\,\mathcal{F}^{\mu\nu} \,\partial_{\xi}A_{\alpha} + 
M_A\,\rho_{\nu}^{\omega}\,\big(\mathcal{C}_{\sigma}^{\sigma\nu}\,\mathcal{L}-\mathcal{C}_{\sigma}^{\sigma\lambda}\, \mathcal{F}_{\lambda\alpha}\,\mathcal{F}^{\alpha\nu}\big). \la{appEq2}
\ee
Substituting these relations into Eq.~\eqref{ELderive1}, we arrive at the identity
\be
\mathcal{E}_{EL}^{\omega}(x) = M_A(x)\,\rho^{\omega}_{\nu}(A)\, \mathcal{E}_{G}^{\nu}(x). \la{ELproofeq}
\ee
Since the matrix $\rho$ is non-degenerate, the equations~\eqref{dMe} and~\eqref{coveq} are equivalent.

\end{proof}}
 {
The integrating factor $M_A(x)$, yielding a gauge-invariant action, is essential for the above proof. Indeed,
it renders the cancellation of the \emph{non} gauge-covariant contributions to $\mathcal{E}_{G}^{\nu}(x)$, which arise from the last line of Eq.~\eqref{appEq1}. 
It is instructive to consider a one-parameter family of integrating factors
\be
M_{A}^Q(x) := \big(M_A(x)\big)^Q,\qquad Q\in\mathbb{R}, 
\ee
which do \emph{not} yield a gauge-invariant action, unless $Q=1$. By using Eq.~\eqref{appEq2} we get:
\bea
&&\frac{1}{M_{A}^Q(x)} \bigg(\frac{\partial M_{A}^Q}{\partial A_{\omega}}\cdot\mathcal{L} - \frac{\partial M_{A}^Q}{\partial x^{\sigma}}\frac{\partial \mathcal{L}}{\partial (\partial_{\sigma} A_{\omega})} \bigg) =\frac{Q}{M_{A}(x)} \bigg(\frac{\partial M_{A}}{\partial A_{\omega}}\cdot\mathcal{L} - \frac{\partial M_{A}}{\partial x^{\sigma}}\frac{\partial \mathcal{L}}{\partial (\partial_{\sigma} A_{\omega})} \bigg)\nonumber\\
&&=  Q\,\rho_{\nu}^{\omega}\,\Big(\rho_{\mu}^{\alpha} \,\mathcal{C}_{\sigma}^{\sigma\xi}\,\mathcal{F}^{\mu\nu} \,\partial_{\xi}A_{\alpha} + 
\big(\mathcal{C}_{\sigma}^{\sigma\nu}\,\mathcal{L}-\mathcal{C}_{\sigma}^{\sigma\lambda}\, \mathcal{F}_{\lambda\alpha}\,\mathcal{F}^{\alpha\nu}\big)\Big),
\eea
and we see that our field equations can be rewritten in the form 
\be
\mathcal{E}_{G}^{\mu}(x,Q) = 0,
\ee
with
\bea
\mathcal{E}_{G}^{\mu}(x,Q) &:=& \frac{1}{M_{A}^Q(x)} \,\big[\rho^{-1}(A)\big]_{\omega}^{\mu}\,\mathcal{E}_{EL}^{\omega}(x)\nonumber\\
&=& \mathcal{D}_{\xi} \mathcal{F}^{\xi\mu} + \frac{1}{2}\, \mathcal{F}_{\lambda\omega}\,\mathcal{C}_{\nu}^{\lambda\omega}\, \mathcal{F}^{\mu\nu} 
 -\mathcal{F}_{\lambda \omega} \, \mathcal{C}_{\nu}^{\mu \omega}\,\mathcal{F}^{\lambda \nu} 
  - \frac{Q}{4}\,\big( \mathcal{C}_{\nu}^{\nu \mu}\, \mathcal{F}_{\lambda\omega}\,\mathcal{F}^{\lambda\omega} 
+ 4\,\mathcal{C}_{\nu}^{\nu \lambda}\,\mathcal{F}_{\lambda \omega}\,\mathcal{F}^{\omega \mu}\big) \nonumber\\
&+&(Q-1)\,\rho_{\lambda}^{\alpha} \,\mathcal{C}_{\sigma}^{\sigma\xi}\,\mathcal{F}^{\lambda\mu} \,\partial_{\xi}A_{\alpha}.
\eea
Though all the terms in the second line of this relation are manifestly gauge-covariant, the last line, which shows up  {for non-unimodular $\mathfrak{g}$} at $Q\neq 1$, breaks the gauge-covariance.  
}

\section{ $\kappa$-Minkowski case} \la{KappaSec}
 {Our results are directly applicable to
the general $\kappa$-Minkowski noncommutativity,
\be
[x^{\mu},x^{\nu}]_{\star} = \ii\,\kappa^{-1} \,(v^{\mu}x^{\nu} -v^{\nu}x^{\mu} ),\qquad v \in \mathbb{R}^{d}.
\ee
The corresponding structure constants are given by
\be
\mathcal{C}^{\mu\nu}_{\sigma} = \kappa^{-1}\,\big(v^{\mu}\, \delta^{\nu}_{\sigma} - v^{\nu}\,\delta^{\mu}_{\sigma}\big),
\ee
and hence
\be
\mathcal{C}^{\sigma\nu}_{\sigma} = - \kappa^{-1}\,(d-1)\,v^{\nu} \neq 0 .
\ee
}

Interestingly, the field equations~\eqref{coveq} have already appeared in the literature,   in the context of the conventional $\kappa$-Minkowski noncommutativity,  $v^{\mu} = \delta_0^{\mu}$,  at $d=4$, see Eq.~(3.4) of~\cite{our2023} at $\alpha=-1/4$. However, an admissible Lagrangian formalism for the $\kappa$-Minkowski case has not yet been developed in~\cite{our2023}, so the mentioned field equations were introduced without any reference to the action principle according to the following guiding lines:
\begin{itemize}
 \item{gauge covariance,} 
 \item{the correct commutative limit,} 
 \item{the existence of a reasonable constraint on the field equations, generalising the Noether identity to the non-Lagrangian setting.}
 \end{itemize}
The last property was established through a direct calculation by using the explicit expressions:
\be
\gamma(A)  =  \left(
\begin{array}{cccc}
1 &-  \frac{A_1}{\kappa} &- \frac{A_2}{\kappa} & -  \frac{A_3}{\kappa} \\
0 &\,\,\,\,\,1 &\,\,\,\,\,0 &\,\,\,\,\,0 \\
0 &\,\,\,\,\,0 &\,\,\,\,\,1 &\,\,\,\,\,0 \\
0 &\,\,\,\,\,0 &\,\,\,\,\,0 &\,\,\,\,\,1 
\end{array}
\right), 
\qquad 
\rho(A)  =  \left(
\begin{array}{cccc}
1 &0 &0 & 0 \\
0 & e^{\frac{A_0}{\kappa} }&0 &0 \\
0 &0 &e^{ \frac{A_0}{\kappa} }&0 \\
0 &0 &0 &e^{ \frac{A_0}{\kappa}} 
\end{array}
\right) \la{newRG}
\ee
for universal-equivalent solutions of the master equations~\eqref{master}.  {We emphasise that the analysis of~\cite{our2023} does not cover other non-unimodular cases, such as the general $\kappa$-Minkowski.}

By substituting~\eqref{newRG} into the definition~\eqref{mudef} of $M_{A}$, we immediately obtain
\be
M_{A}(x) = \exp{\big(-3\kappa^{-1} A_0(x) \big)}, \la{kappaIF}
\ee
so the corresponding action becomes
\be
S_{\mathbf{g}}[A] = \int_{\mathcal{M}}\dd x\, \exp\big(-3\kappa^{-1} A_0(x)\big )\bigg(- \frac{1}{4}\, \mathcal{F}_{\mu\nu}(x)\, \mathcal{F}^{\mu\nu}(x)  \bigg). \la{kappaAction}
\ee
 {Remarkably, at the leading order in $\kappa^{-1}$, the integrating factor~\eqref{kappaIF} perfectly agrees with the field-dependent volume factor proposed in~\cite{Dimitrijevic:2005xw} within the perturbative first-order approach to noncommutative electrodynamics on the $\kappa$-Minkowski spacetime, see Eqs.~(4.12) and~(4.14) of that reference.

Since, as established, } Eq.~(3.4) of~\cite{our2023} at $\alpha=-1/4$ can be obtained from the action principle, the corresponding constraint (Eq.(3.5) of~\cite{our2023} at $\alpha=-1/4$)
\be
\mathcal{D}_{\mu}\,\mathcal{E}^{\mu}_{G} = -\,\mathcal{C}_{\xi}^{\mu\nu}\,\mathcal{F}_{\mu\nu}\,\mathcal{E}^{\xi}_{G}\,,
\ee
is a true Noether identity, representing Eq.~\eqref{gind} in a manifestly gauge-covariant form.

Of course, one may wonder whether the field equations~(3.4) of~\cite{our2023} for other values of the parameter $\alpha$  can be obtained from the action~\eqref{invaction} for some non-minimal choice of the Lagrangian density, cf. Remark~1 after Proposition~\ref{prop5}. We do not exclude this possibility; however, such an analysis goes beyond the scope of the present paper.

\section{Summary and perspectives}
We have addressed the problem of finding a local gauge-invariant classical action for a generic Lie-Poisson electrodynamics.  
The most important results of the paper are given by Propositions~\ref{prop4}  and~\ref{propMaxwell}, which provide the action and the corresponding deformed Maxwell equations for any Lie-Poisson gauge model.
These propositions generalize the admissible Lagrangian models of~\cite{our2023} to non-unimodular Lie-algebra-type noncommutativities which include the  $\kappa$-Minkowski case. 

 Technically, the main novelty compared to our previous study~\cite{our2023} is the field-dependent integrating factor~\eqref{mudef}, which enables us to construct the action~\eqref{invaction}, valid for all Lie-algebra-type noncommutativities. For the universal and universal-equivalent realizations of Lie-Poisson electrodynamics, we have calculated this integrating factor explicitly, see Eqs.~\eqref{mufirs} and~\eqref{musec}, respectively.  
At the semiclassical level,~\eqref{mudef} generalizes the volume factor proposed in~\cite{Dimitrijevic:2005xw}, extending it to arbitrary order in the deformation parameter and to arbitrary Lie-algebra-type noncommutativity. 

It is worth noticing that an alternative approach~\cite{Kupriyanov:2023qot}, based on a different deformed field strength, does not require any special integrating factor. However, the local action (4.10) of that reference introduces a metric tensor $g_{\mu\nu}^{\Sigma}$ depending on the gauge field $A$ and its first derivatives $\partial A$. This tensor has been computed explicitly only for the case of constant noncommutativity, while for Lie-algebra-type noncommutativities the construction provides only a set of prescriptions. In particular, explicit deformed Maxwell equations are not obtained within that framework.

We applied our machinery to the four-dimensional $\kappa$-Minkowski case and obtained a quite simple, albeit nontrivial, gauge-invariant action~\eqref{kappaAction}, thereby giving a Lagrangian formulation to the field equations~(3.4) of~\cite{our2023} at $\alpha=-1/4$,  proposed in the quoted paper on general grounds  without any reference to the action principle.  
We also established the \emph{universality} of these deformed Maxwell equations. While, as explained in Sec.~\ref{KappaSec}, the analysis of~\cite{our2023} was performed for a four-dimensional $\kappa$-Minkowski model only, in Sec.~\ref{action} we have shown that the field equations~\eqref{coveq} are the true Euler-Lagrange equations for a generic Lie-Poisson electrodynamics.

The Lagrangian formulation opens further prospects for a Hamiltonian analysis and subsequent canonical quantization of the model.  {Apart from that, by combining our results with those of~\cite{Basilio:2024bir}, one can obtain the deformed Maxwell equations in the presence of a charged particle's current, cf. Remark~3 after Proposition~\ref{prop5}. These nonhomogeneous equations would allow an analysis of the electromagnetic radiation by moving particles in the Lie-Poisson setting. In particular, it would be interesting to study deformed Liénard-Wiechert potentials and bremsstrahlung.} And finally, our results point at the right direction for searching for a classical action invariant under the full algebra~\eqref{ncalg} of noncommutative $U(1)$ gauge transformations in the $\kappa$-Minkowski and other cases of non-cyclic star-products. Some steps beyond the semiclassical approximation were already taken perturbatively in~\cite{Abla:2022wfz} within the $L_{\infty}$ formalism.

 One of the most important open questions in the Lie-Poisson gauge theory concerns the deformed space-time symmetries. While all the research of Refs.~\cite{our2023,Kupriyanov:2019cug,Kurkov:2021kxa,Kupriyanov:2020axe,Abla:2022wfz,Kupriyanov:2021aet,Kupriyanov:2021cws,Kupriyanov:2022ohu,Kupriyanov:2023qot,DiCosmo:2023wth,Bascone:2024mxs,Kupriyanov:2024dny,Sharapov:2024bbu, Basilio:2024bir,Abla:2024wtr} is based on the deformed gauge algebra~\eqref{pga}, the deformed Poincaré transformations have never been analysed. The Poisson bracket~\eqref{Pbre} entering both the gauge algebra and the deformed field strength clearly breaks the usual Poincaré symmetry, however, one may look for the \emph{deformed} space-time transformations. At this point, the above link between the present study and that of~\cite{Dimitrijevic:2005xw} gives us optimism: the latter exhibits the $\kappa$-Poincaré invariance. A deeper analysis and comparison of the two approaches may hint towards the right direction for the search of the deformed space-time symmetries in the Lie-Poisson formalism. We hope to address this problem in future publications. 

\begin{appendix}
 {
\section{Technicalities} \la{appA}
Throughout this appendix, we use the short-hand notation 
$
\partial_A^{\mu} = \frac{\partial}{\partial A_{\mu}}
$
and omit the argument $A(x)$ of the matrices $\gamma$ and $\rho$ whenever it does not cause confusion.

\subsection{Derivation of Eq.~\eqref{appEq1}}
A straightforward calculation yields:
\bea
\frac{\partial {\mathcal{F}_{\mu\nu}}}{\partial A_{\omega}} 
&=& \big(\partial_A^{\omega} \rho_{\mu}^{\alpha}\big) \, \rho_{\nu}^{\beta}\, \hat{\mathcal{F}}_{\alpha\beta}
+\rho_{\nu}^{\beta}\, \rho_{\mu}^{\alpha} \,\big(\partial_{A}^{\omega}\gamma_{\alpha}^{\xi}\big)\, \partial_{\xi} A_{\beta} - (\mu\longleftrightarrow \nu) \nonumber\\
&=& \rho_{\nu}^{\beta}\,\Big(\big(\partial_A^{\omega}\rho_{\mu}^{\alpha}-\partial_A^{\alpha}\rho_{\mu}^{\omega}\big)\,\hat{ {\mathcal{F}}}_{\alpha\beta}
+\{\rho_{\mu}^{\omega},A_{\beta}\} - \gamma^{\xi}_{\beta}\,\partial_{\xi}\rho_{\mu}^{\omega}\Big) - (\mu\longleftrightarrow \nu),
\eea
where we introduced the notation
\be
\hat{\mathcal{F}}_{\xi\lambda} := \gamma_{\xi}^{\sigma}(A)\,\partial_{\sigma}A_{\lambda}
   -\gamma_{\lambda}^{\sigma}(A)\,\partial_{\sigma}A_{\xi}+\{A_{\xi},A_{\lambda}\} \la{hatFdef}
\ee
and applied the second master equation~\eqref{master} to $ \rho_{\mu}^{\alpha} \,\big(\partial_{A}^{\omega}\gamma_{\alpha}^{\xi}\big)$.  According to Proposition~2.1 of~\cite{our2023},
\be
\partial_A^{\omega}\rho_{\mu}^{\alpha}-\partial_A^{\alpha}\rho_{\mu}^{\omega} = \rho_{\sigma}^{\omega}\,\rho_{\varepsilon}^{\alpha}\,
\mathcal{C}_{\mu}^{\sigma\varepsilon}, \la{prop21our2023}
\ee
therefore,
\be
\frac{\partial {\mathcal{F}_{\mu\nu}}}{\partial A_{\omega}}  
= \rho_{\sigma}^{\omega}\,\mathcal{C}_{\mu}^{\sigma\varepsilon}\,\mathcal{F}_{\varepsilon\nu} 
+ \rho_{\nu}^{\beta}\,\big(\{\rho_{\mu}^{\omega},A_{\beta}\} - \gamma_{\beta}^{\xi}\,\partial_{\xi}\rho_{\mu}^{\omega}\big) - (\mu\longleftrightarrow\nu),
\la{firsder1}
\ee
where we took into account that,
\be
\mathcal{F}_{\varepsilon\nu} = \rho_{\varepsilon}^{\alpha}\,\rho_{\nu}^{\beta}\,\hat{\mathcal{F}}_{\alpha\beta}. \la{hatFFrel}
\ee
The relation~\eqref{firsder1} yields the following expression for the first term of the l.h.s. of Eq.~\eqref{appEq1}: 
\bea
\frac{\partial\mathcal{L}}{\partial A_{\omega}} = -\frac{1}{2}\,\frac{\partial\mathcal{F}_{\mu\nu}}{\partial A_{\omega}}\,\mathcal{F}^{\mu\nu}
= -\rho_{\sigma}^{\omega}\,\mathcal{C}_{\mu}^{\sigma\varepsilon}\,\mathcal{F}_{\varepsilon\nu}\,\mathcal{F}^{\mu\nu} 
+\rho_{\nu}^{\beta}\,\big(\{A_{\beta},\rho_{\mu}^{\omega}\} + \gamma_{\beta}^{\xi}\,\partial_{\xi}\rho_{\mu}^{\omega}\big)\,\mathcal{F}^{\mu\nu}. \la{firsder2}
\eea

To analyze the second term of the l.h.s. of Eq.~\eqref{appEq1}, one can easily see that:
\be
\frac{\partial \mathcal{F}_{\mu\nu}}{\partial(\partial_{\sigma}A_{\omega})} = \rho_{\nu}^{\omega}\,\rho_{\mu}^{\alpha}\,\big(\gamma^{\sigma}_{\alpha} + 
x^{\beta}\,\mathcal{C}^{\xi\sigma}_{\beta}\,\partial_{\xi}A_{\alpha}\big) - (\mu\longleftrightarrow\nu). \la{FmunuAos}
\ee
By noticing that
\be
\frac{1}{2}\,\frac{\partial \mathcal{F}_{\mu\nu}}{\partial(\partial_{\sigma}A_{\omega})}\,\partial_{\sigma}\mathcal{F}^{\mu\nu} = 
\rho_{\nu}^{\omega} \,\mathcal{D}_{\mu}\mathcal{F}^{\mu\nu},
\ee
and renaming the mute indices in a suitable way, we obtain: 
\bea
\!\!\!\!\!\!\!\!-\frac{\partial}{\partial x^{\sigma}} \frac{\partial \mathcal{L}}{\partial(\partial_{\sigma}A_{\omega})} 
&=& \frac{1}{2}\,\partial_{\sigma}\bigg(\frac{\partial \mathcal{F}_{\mu\nu}}{\partial(\partial_{\sigma}A_{\omega})}\,\mathcal{F}^{\mu\nu}\bigg) \nonumber\\
&=&\rho_{\sigma}^{\omega} \,
\mathcal{D}_{\mu}\mathcal{F}^{\mu\sigma} - \big(\partial_{\xi}\big(\rho^{\beta}_{\nu}\,\gamma^{\xi}_{\beta}\,\rho_{\mu}^{\omega}\big) 
+ \{A_{\beta},\rho_{\mu}^{\omega}\,\rho_{\nu}^{\beta}\}
-\rho^{\omega}_{\mu}\,\rho_{\nu}^{\beta}\,\mathcal{C}^{\sigma\xi}_{\sigma}\,\partial_{\xi}A_{\beta}\big)\,\mathcal{F}^{\mu\nu} . {\la{firsder3}}
\eea

By combining the equations~\eqref{firsder2} and~\eqref{firsder3}, we get:
\bea
\frac{\partial\mathcal{L}}{\partial A_{\omega}}-\frac{\partial}{\partial x^{\sigma}} \frac{\partial \mathcal{L}}{\partial(\partial_{\sigma}A_{\omega})}&=&
\rho_{\sigma}^{\omega}\,\big(\mathcal{D}_{\mu}\mathcal{F}^{\mu\sigma}-\,\mathcal{C}_{\mu}^{\sigma\varepsilon}\,\mathcal{F}_{\varepsilon\nu}\,\mathcal{F}^{\mu\nu} \big)+\rho^{\omega}_{\mu}\,\rho_{\nu}^{\beta}\,\mathcal{C}^{\sigma\xi}_{\sigma}\,\partial_{\xi}A_{\beta}\,\mathcal{F}^{\mu\nu} \nonumber\\
&-&\rho_{\sigma}^{\omega}\,\big(\partial_{\xi}\big(\rho^{\beta}_{\nu}\,\gamma^{\xi}_{\beta}\big) 
+ \{A_{\beta},\rho_{\nu}^{\beta}\} \big)\,\mathcal{F}^{\sigma\nu}. \la{firsder4}
\eea
The straightforward calculation yields:
\bea
\partial_{\xi}\big(\rho^{\beta}_{\nu}\,\gamma^{\xi}_{\beta}\big) 
+ \{A_{\beta},\rho_{\nu}^{\beta}\} &=& \big(\gamma^{\xi}_{\beta}\,\partial_{A}^{\alpha}\rho^{\beta}_{\nu}+\rho^{\beta}_{\nu}\,\partial_{A}^{\alpha}\gamma^{\xi}_{\beta}\big)\,\partial_{\xi}A_{\alpha} -\partial_{A}^{\alpha}\rho_{\nu}^{\beta}\,\{A_{\alpha},A_{\beta}\} \nonumber\\
&=& \big(\partial_{A}^{\alpha}\rho^{\beta}_{\nu} - \partial_{A}^{\beta}\rho^{\alpha}_{\nu}\big)\,\gamma^{\xi}_{\beta}\,\partial_{\xi}A_{\alpha} -\partial_{A}^{\alpha}\rho_{\nu}^{\beta}\,\{A_{\alpha},A_{\beta}\}\nonumber\\
&=&-\frac{1}{2}\underbrace{\big(\partial_{A}^{\alpha}\rho^{\beta}_{\nu} - \partial_{A}^{\beta}\rho^{\alpha}_{\nu}\big)}_{\mathcal{C}_{\nu}^{\lambda\phi}\rho^{\alpha}_{\lambda}\,\rho^{\beta}_{\phi}}\underbrace{\big(\gamma^{\xi}_{\alpha}\,\partial_{\xi}A_{\beta} -\gamma^{\xi}_{\beta}\,\partial_{\xi}A_{\alpha}+
\{A_{\alpha},A_{\beta}\}\big)}_{\hat{\mathcal{F}}_{\alpha\beta}}\nonumber\\
&=&-\frac{1}{2}\,\mathcal{C}_{\nu}^{\lambda\phi}\,\mathcal{F}_{\lambda\phi}, \la{longinterme}
\eea
where we applied the second master equation~\eqref{master} to the combination $\rho^{\beta}_{\nu}\,\partial_{A}^{\alpha}\gamma^{\xi}_{\beta}$ and used the relations~\eqref{prop21our2023}, \eqref{hatFdef} and~\eqref{hatFFrel} at the last step. Substituting the equality~\eqref{longinterme} in Eq.~\eqref{firsder4},
we arrive at the desired relation~\eqref{appEq1}.

\subsection{Derivation of Eq.~\eqref{appEq2}}
First we prove two useful formulae, which make the required derivation straightforward. 
\begin{proposition}\label{propA1}
The integrating factor $M_A$, and the matrices $\rho(A)$ and $\gamma(A)$ obey the following equations: 
\bea
\frac{\partial M_A}{\partial A_{\omega}} = M_A\,\mathcal{C}^{\beta\varepsilon}_{\beta}\,\rho_{\varepsilon}^{\omega}(A), \la{eA1}
\eea
and
\be
\mathcal{C}^{\beta\varepsilon}_{\beta}\,\rho_{\varepsilon}^{\lambda}(A)\,\gamma_{\lambda}^{\sigma}(A) = \mathcal{C}_{\beta}^{\beta\sigma}. \la{eA2}
\ee
\end{proposition}
\begin{proof}
In analogy with the derivation of~\eqref{todemost}, by using Eq.~\eqref{interme3} along with 
\be
\big[\rho_{\mathbf{u}}\big]^{\xi}_{\sigma}(A) = \delta^{\xi}_{\sigma}+ \sum_{k=1}^{\infty}\frac{\big[\hat{A}^k\big]^{\xi}_{\sigma}}{(k+1)!}, \la{rhoexp}
\ee
we arrive at the following peculiar property of the universal realization:
\be
 \mathcal{C}_{\nu}^{\nu\xi} \,\big[\rho_{\mathbf{u}}\big]_{\xi}^{\sigma}
 =    \mathcal{C}_{\nu}^{\nu\sigma}. \la{todemost2}
\ee
Equations~\eqref{todemost2}, \eqref{todemost} and~\eqref{mufirs} immediately yield the desired relations~\eqref{eA1} and~\eqref{eA2} for the universal case,
thus in what follows we focus on the universal-equivalent realization.

To prove ~\eqref{eA1}, we start from the explicit expression~\eqref{musec} and then apply Eq.~\eqref{todemost2} for 
$\rho_{\mathbf{u}}\big(\tilde{A}(A)\big)$:
\be
\frac{\partial M_A}{\partial A_{\omega}}  = M_A\,\mathcal{C}^{\beta\varepsilon}_{\beta}\, \frac{\partial \tilde{A}_{\varepsilon}}{\partial A_{\omega}}
=  M_A\,\mathcal{C}^{\beta\theta}_{\beta}\,\underbrace{\big[\rho_{\mathbf{u}}\big]_{\theta}^{\varepsilon}\big(\tilde{A}(A)\big)\, \frac{\partial \tilde{A}_{\varepsilon}}{\partial A_{\omega}}}_{\rho_{\theta}^{\omega}(A)} = M_A\,\mathcal{C}^{\beta\theta}_{\beta}\,\rho_{\theta}^{\omega}(A),
\ee 
where we used~\eqref{newGR} at the last step.

The second equality is a direct consequence of Eq.~\eqref{secsteinterme} along with Eq.~\eqref{todemost2} and Eq.~\eqref{todemost}, applied for $\gamma_{\mathbf{u}}\big(\tilde{A}(A)\big)$ and $\rho_{\mathbf{u}}\big(\tilde{A}(A)\big)$:
\be
\mathcal{C}^{\beta\varepsilon}_{\beta}\,\rho_{\varepsilon}^{\lambda}(A)\,\gamma_{\lambda}^{\sigma}(A) 
=\underbrace{\mathcal{C}^{\beta\varepsilon}_{\beta}\,\big[\rho_{\mathbf{u}}\big]_{\varepsilon}^{\lambda}\big(\tilde{A}(A)\big)}_{\mathcal{C}^{\beta\lambda}_{\beta}}\,\big[\gamma_{\mathbf{u}}\big]_{\lambda}^{\sigma}\big(\tilde{A}(A)\big) 
= \mathcal{C}_{\beta}^{\beta\sigma}.
\ee 
\end{proof}
By using Eq.~\eqref{eA1} along with Eq.~\eqref{FmunuAos}, for the left-hand side of~\eqref{appEq2} we obtain:
\be
\frac{\partial M_A}{\partial A_{\omega}}\cdot\mathcal{L} - \frac{\partial M_A}{\partial x^{\sigma}}\frac{\partial \mathcal{L}}{\partial (\partial_{\sigma} A_{\omega})}
=M_A\,\rho_{\nu}^{\omega}\,\Big(\mathcal{C}^{\beta\nu}_{\beta}\,\mathcal{L}
+ \mathcal{C}_{\beta}^{\beta\varepsilon}\,\rho_{\mu}^{\alpha}\,\rho_{\varepsilon}^{\lambda}\,\big(\gamma_{\alpha}^{\sigma}\,\partial_{\sigma}A_{\lambda}
+\{A_{\alpha},A_{\lambda}\}\big)\,\mathcal{F}^{\mu\nu} \Big)  \la{secoderinterme}
\ee
By adding and subtracting a suitable term, we complete our derivation of~Eq.~\eqref{appEq2}:
\bea
 \mathbf{r.h.s.}~\mbox{Eq.~\eqref{secoderinterme}}&=& M_A\,\rho_{\nu}^{\omega}\,\Big(\mathcal{C}^{\beta\nu}_{\beta}\,\mathcal{L}
+ \mathcal{C}_{\beta}^{\beta\varepsilon}\,\underbrace{\rho_{\mu}^{\alpha}\,\rho_{\varepsilon}^{\lambda}\,\big(\gamma_{\alpha}^{\sigma}\,\partial_{\sigma}A_{\lambda}
-\gamma_{\lambda}^{\sigma}\,\partial_{\sigma}A_{\alpha}
+\{A_{\alpha},A_{\lambda}\}\big)}_{\mathcal{F}_{\mu\varepsilon}}\,\mathcal{F}^{\mu\nu} \Big) \nonumber\\
&+& M_A\,\rho_{\nu}^{\omega}\,\rho_{\mu}^{\alpha}\,\underbrace{\mathcal{C}_{\beta}^{\beta\varepsilon}\,\rho_{\varepsilon}^{\lambda}\,\gamma_{\lambda}^{\sigma}}_{\mathcal{C}^{\beta\sigma}_{\beta}}\,\partial_{\sigma}A_{\alpha}
\,\mathcal{F}^{\mu\nu}  \nonumber\\
&=&M_A\,\rho_{\nu}^{\omega}\,\Big(\mathcal{C}^{\beta\nu}_{\beta}\,\mathcal{L}
- \mathcal{C}_{\beta}^{\beta\varepsilon}\,\mathcal{F}_{\varepsilon\mu}\,\mathcal{F}^{\mu\nu} \Big)
+ M_A\,\rho_{\nu}^{\omega}\,\rho_{\mu}^{\alpha}\,\mathcal{C}^{\beta\sigma}_{\beta}\,\partial_{\sigma}A_{\alpha}
\,\mathcal{F}^{\mu\nu} 
,
\eea
where we applied Eq.~\eqref{eA2} at the last step.
}
\end{appendix}

{\section*{Acknowledgements} M.K. is grateful to Vlad Kupriyanov, Fedele Lizzi, Alexey Sharapov and Patrizia Vitale for fruitful discussions on the present and related papers.}


\end{document}